\documentclass[11pt]{article}

\usepackage[draft]{yuhao-template}

\newcommand{\MMS}{\textnormal{MMS}}

\newcommand{\vr}{\mathbf{r}}

\newcommand{\bv}{p}

\title{Exact MMS Allocations under Personalized Bivalued Valuations:\\
Goods and Chores}
\author{Yuhao Zhang\\
Shanghai Jiao Tong University\\
\texttt{zhang\_yuhao@sjtu.edu.cn}}
\date{\today}

\begin{document}

\maketitle

\begin{abstract}
The maximin share (MMS) is a central fairness benchmark for allocating indivisible goods and chores. We study additive valuations in the personalized bivalued setting, where each agent assigns one of two agent-specific values to every item. Whether exact MMS allocations always exist in this setting has remained a major open question, as highlighted by Ebadian, Peters, and Shah and by Garg, Huang, and Segal-Halevi. We answer this question affirmatively: we prove that exact MMS allocations always exist for both goods and chores and can be computed in polynomial time. Our proof combines a quota-based reformulation with an envelope relaxation, a sparse extreme-point construction, and flow-based rounding that controls the total rounding loss.
\end{abstract}

\section{Introduction}
\label{sec:introduction}

Fair division asks how to allocate indivisible items among agents with heterogeneous preferences. We consider $m$ items and $n$ agents, and require every item to be allocated. Items may be goods, which agents prefer to receive, or chores, which agents prefer to avoid. Since agents may evaluate the same item differently, fairness must be defined relative to each agent's own valuation.

Two prominent approaches are envy-based and share-based fairness. Envy-based notions compare an agent's bundle with those assigned to others; examples include envy-freeness and its relaxations EF1 and EFX\@. Share-based notions instead compare each agent's bundle with a benchmark derived from her own valuation. Proportionality, for example, asks each agent to receive at least $1/n$ of her total utility for goods or to incur at most $1/n$ of her total cost for chores. Even for goods, proportionality may be impossible with indivisible items: if a single good is valuable to every agent, at most one can receive it, leaving the others with zero value despite positive proportional benchmarks.

\citet{DBLP:conf/bqgt/Budish10} introduced the maximin share (MMS) guarantee for indivisible goods as a more robust share-based benchmark. Agent $i$ partitions the items into $n$ bundles and assumes that she receives the least valuable one; her MMS is the largest value she can guarantee by choosing the partition. The corresponding notion for chores was first studied by \citet{DBLP:conf/aaai/AzizRSW17}; it is the smallest worst-bundle cost that an agent can guarantee. An exact MMS allocation assigns to each agent a bundle that meets the corresponding benchmark. Unfortunately, exact MMS allocations need not exist in general, for either goods \citep{DBLP:journals/jacm/KurokawaPW18} or chores \citep{DBLP:conf/aaai/AzizRSW17}. However, the known counterexamples rely on carefully constructed valuation profiles, leaving open whether exact MMS allocations exist for structured and practically relevant valuation classes.

We focus on \emph{personalized bivalued valuations}, a simple model that retains the central heterogeneity of fair division while supporting coarse preference elicitation. Each agent classifies goods as preferred or nonpreferred, or chores as easy or difficult, and assigns one value to each category. The two values may differ across agents.

Previous work either imposes additional restrictions on personalized bivalued valuations or relaxes the MMS benchmark\@. \citet{feige2022maximin} assumed that all agents share a common value pair, that is, the non-personalized bivalued setting, and proved exact MMS existence for both goods and chores. In a separate result, \citet{DBLP:conf/atal/EbadianP022} allowed the value pair to vary across agents and obtained the same conclusion when each pair can be normalized to $\{1,p_i\}$ for an integer $p_i$.

Without either restriction, prior work has instead studied multiplicative relaxations of MMS. For personalized bivalued chores, \citet{DBLP:conf/aaai/GargHS25} obtained a $15/13$-MMS guarantee. To our knowledge, personalized bivalued goods have not received a stronger class-specific approximation; the best available guarantee is inherited from general additive valuations, for which \citet{DBLP:journals/corr/abs-2511-13056} recently obtained a $7/9$-MMS allocation. Whether exact MMS allocations exist for personalized bivalued valuations remained a major open question, highlighted in \citet{DBLP:conf/atal/EbadianP022} and \citet{DBLP:conf/aaai/GargHS25}. We resolve this question for both goods and chores.

\begin{center}
\setlength{\fboxsep}{8pt}%
\fbox{%
\begin{minipage}{0.9\linewidth}
\textbf{Main result.}
Every personalized bivalued instance of indivisible goods or indivisible chores admits an exact MMS allocation. Moreover, such an allocation can be computed in polynomial time.
\end{minipage}}
\end{center}

\subsection{Technical Highlights}
The proofs for goods and chores share the same framework. We present the goods argument; the chores proof uses the corresponding capacity-based construction, with an upper concave envelope in place of a lower convex envelope.

\paragraph{Quota-based reformulation and $n$ worlds.}
Our first idea is a quota-based reformulation of the MMS constraints. After normalizing agent $i$'s two values to $1$ and $p_i>1$, we let a preferred quota $r$ specify how many preferred goods she should receive. Her demand $f_i(r)$ is the minimum total number of goods sufficient to meet her MMS if at least $r$ of them are preferred. Thus, the problem reduces to finding an implementable preferred-quota vector. Rather than construct a single vector directly, we consider $n$ candidate vectors, viewed as $n$ different worlds, and distribute each agent's preferred quotas fractionally across them. A prefix-capacity condition ensures that every world respects the available supply of preferred goods, while an aggregate bound controls the total demand across all worlds. A best-of-$n$-worlds argument will then identify one implementable world.

\paragraph{Fractional feasibility and integral tightness.}
The ceiling in the original demand function makes even fractional feasibility nontrivial: the natural uniform fractional table need not satisfy the desired demand bound. We therefore replace $f_i$ by its lower convex envelope $g_i$, the convex polygonal curve drawn below the integral demand points and pushed upward as far as possible. This choice serves two purposes. First, pushing the curve upward makes it pass through as many original integral demand points as convexity permits. At each such tight quota, $g_i(r)=f_i(r)$, so no loss occurs when we later return from the relaxed demand to the original demand. Second, convexity, together with $g_i(r)\leq f_i(r)$ at every integral quota, ensures that the relaxed demand at an average quota is at most the average original demand of the integral quotas being mixed. Applying this inequality to the quotas induced by an MMS partition shows that the uniform $n$-world table satisfies the relaxed-demand bound.

\paragraph{Sparse extreme points and low-loss rounding.}
For each agent, we restrict her quotas to the affine segment of $g_i$ containing the uniform quota. We call an entry a bound entry if it lies at an endpoint of this segment. As observed above, both endpoints are integral tight quotas; hence, every bound entry is already integral, satisfies $g_i=f_i$, and incurs no rounding loss. Within these segment restrictions, we further restrict the feasible tables to a weighted transportation polytope and choose an extreme point. The key intuition, going back to the classical transportation-basis argument of \citet{dantzig1951application}, is that a cycle of non-bound entries admits an alternating two-sided perturbation and therefore cannot occur at an extreme point. Thus, the non-bound entries form a forest and constitute only a sparse set. We then use a max-flow argument to round these entries while preserving all per-agent quota totals and per-world prefix-capacity constraints. Each non-bound entry contributes less than one unit of additional demand, and the forest structure limits the aggregate rounding loss to $n-1$. The best-of-$n$-worlds argument therefore yields an implementable integral quota vector and hence an exact MMS allocation.

\subsection{Other Related Work}

\paragraph{Goods.}
For general additive goods, exact MMS allocations need not exist \citep{DBLP:journals/jacm/KurokawaPW18}. Starting from the $2/3$-MMS guarantee of \citet{DBLP:journals/jacm/KurokawaPW18}, the approximation factor was improved to $3/4$ \citep{DBLP:conf/sigecom/GhodsiHSSY18} and subsequently to $3/4+o(1)$ \citep{DBLP:journals/ai/GargT21,DBLP:conf/ijcai/AkramiGST23}. The first constant improvement beyond $3/4$ was $3/4+3/3836$ \citep{DBLP:conf/soda/AkramiG24}, followed by $10/13$ \citep{DBLP:conf/soda/HeidariKSS26} and the current best factor of $7/9$ \citep{DBLP:journals/corr/abs-2511-13056}. On the negative side, no factor above $39/40$ is possible even for three agents \citep{DBLP:conf/wine/FeigeST21}. Other structured valuation classes admitting exact MMS allocations include ternary valuations with values in $\{0,1,2\}$ \citep{DBLP:journals/talg/AmanatidisMNS17} and weakly lexicographic valuations \citep{DBLP:conf/atal/EbadianP022}.

\paragraph{Chores.}
For general additive chores, \citet{DBLP:conf/aaai/AzizRSW17} showed that exact MMS allocations may fail and gave a polynomial-time $2$-approximation. The approximation factor was subsequently improved to $4/3$ \citep{DBLP:journals/teco/BarmanK20}. \citet{DBLP:conf/sigecom/HuangL21} then obtained an existential guarantee of $11/9$ and a polynomial-time guarantee of $5/4$. The current best general factor is $13/11$, together with a polynomial-time $(13/11+\varepsilon)$-approximation for every $\varepsilon>0$ \citep{DBLP:conf/sigecom/HuangS23}. On the negative side, no factor below $44/43$ is possible even for three agents \citep{DBLP:conf/wine/FeigeST21}. Other structured valuation classes admitting exact MMS allocations include weakly lexicographic valuations \citep{DBLP:conf/atal/EbadianP022} and factored instances \citep{DBLP:conf/aaai/GargHS25}.

\subsection{Organization}
We begin in \Cref{sec:preliminaries} by formalizing the model and recording some basic reductions used throughout the paper. For goods, we first give a proof overview in \Cref{sec:goods-overview}, then develop the fractional relaxation in \Cref{sec:fractional-demand}, construct a sparse extreme-point solution in \Cref{sec:extreme-point-solution}, and finally give the low-loss rounding procedure in \Cref{sec:prefix-rounding} and prove the following theorem.

\begin{restatable}{theorem}{exactMMSGoodsTheorem}
\label{thm:exact_mms_goods}
Every personalized bivalued instance of indivisible goods admits an exact MMS allocation. Moreover, such an allocation can be computed in polynomial time.
\end{restatable}

We then turn to chores in \Cref{sec:chores}, where we adapt the same framework to capacities and prove the following theorem.

\begin{restatable}{theorem}{exactMMSChoresTheorem}
\label{thm:exact_mms_chores}
Every personalized bivalued instance of indivisible chores admits an exact MMS allocation. Moreover, such an allocation can be computed in polynomial time.
\end{restatable}

\section{Preliminaries}
\label{sec:preliminaries}

For every nonnegative integer $q$, let $[q]:=\{1,\ldots,q\}$, with $[0]:=\emptyset$. There are $n\geq 1$ agents, indexed by $[n]$, and $m\geq 0$ indivisible items, indexed by $[m]$. Each agent--item pair $(i,j)\in[n]\times[m]$ is associated with a value $v_{ij}$. We consider two variants of the problem, depending on whether the items are goods or chores. For every agent $i\in[n]$, the valuation is additive; that is, for every bundle $S\subseteq[m]$, we have
$$
v_i(S)=\sum_{j\in S}v_{ij}.
$$
All values are nonnegative rationals represented in binary. This paper focuses on the \emph{personalized bivalued} setting: for each agent $i\in[n]$, there are two agent-specific values $0\leq a_i<b_i$ such that $v_{ij}\in\{a_i,b_i\}$ for every item $j\in[m]$. We do not require both values to occur, so a constant valuation can be represented by choosing an arbitrary distinct value that is not realized. No relationship is imposed between the value pairs $(a_i,b_i)$ of different agents. The main proofs initially assume that all realized item values are strictly positive. This is without loss of generality: the zero-value cases for both goods and chores are handled in \Cref{sec:zero-lower-values}.

In the goods setting, $v_{ij}$ is the utility that agent $i$ derives from item $j$. In the chores setting, $v_{ij}$ is the disutility, or cost, incurred by agent $i$ for performing chore $j$.

An allocation is a tuple $A=(A(1),\ldots,A(n))$ that forms a partition of $[m]$. In other words, the bundles $A(1),\ldots,A(n)$ are pairwise disjoint and satisfy $\bigcup_{i\in[n]}A(i)=[m]$. Under allocation $A$, agent $i$ receives bundle $A(i)$, whose value to her is
$$
v_i(A(i))=\sum_{j\in A(i)}v_{ij}.
$$

Let $\mathcal{A}$ denote the set of all allocations. In the goods setting, the \emph{maximin share} of agent $i$ is
$$
\MMS_i
=
\max_{A\in\mathcal{A}}
\min_{i'\in[n]}v_i(A(i')).
$$
Thus, agent $i$ partitions the goods into $n$ bundles and then receives the bundle that she values the least. Her maximin share $\MMS_i$ is the largest utility that she can guarantee by choosing the partition. An allocation $A\in\mathcal{A}$ is an \emph{exact MMS allocation for goods} if
$$
v_i(A(i))\geq\MMS_i
\qquad
\text{for every }i\in[n].
$$

In the chores setting, the corresponding benchmark is the \emph{minimax share}
$$
\MMS_i
=
\min_{A\in\mathcal{A}}
\max_{i'\in[n]}v_i(A(i')).
$$
Here, agent $i$ partitions the chores into $n$ bundles and then receives the bundle with the largest cost according to her valuation. Hence, $\MMS_i$ is the smallest worst-case cost that she can guarantee by choosing the partition. An allocation $A\in\mathcal{A}$ is an \emph{exact MMS allocation for chores} if
$$
v_i(A(i))\leq\MMS_i
\qquad
\text{for every }i\in[n].
$$
Although the two benchmarks optimize in opposite directions, we refer to both as the MMS benchmark and use the same notation $\MMS_i$.

\paragraph{Identical orderings.}
An instance has an \emph{identical ordering} (IDO) if the items can be indexed so that
$$
v_{i1}\geq v_{i2}\geq\cdots\geq v_{im}
\qquad
\text{for every }i\in[n].
$$
The reduction to IDO instances was introduced by \citet{DBLP:journals/aamas/BouveretL16} for goods and first adapted to chores by \citet{DBLP:conf/sigecom/HuangL21}. In both settings, the reduction preserves every agent's MMS, and any exact MMS allocation of the resulting IDO instance can be converted in polynomial time into an exact MMS allocation of the original instance. Thus, existence and polynomial-time computation for IDO instances imply the same results for arbitrary instances, and we restrict attention to IDO instances throughout the technical sections.

\section{Proof Overview for Goods}
\label{sec:goods-overview}

The proofs for goods and chores are parallel. We present the goods argument in the main text and its chore counterpart in \Cref{sec:chores}. This section outlines the goods proof.

\paragraph{Cardinality reformulation of the MMS constraints.}
The first step is to reformulate the MMS requirement in the personalized bivalued IDO setting. Recall that, in an IDO instance, the goods are indexed in nonincreasing order of value for every agent. Because MMS is homogeneous under positive scaling, we normalize each agent's nonpreferred value to $1$. Consequently, for every agent $i\in[n]$, there exist a preferred value $\bv_i>1$ and a threshold $\theta_i$ such that
$$
v_{ij}
=
\begin{cases}
\bv_i, & \text{if } j\leq\theta_i,\\
1, & \text{if } j>\theta_i.
\end{cases}
$$
Thus, the first $\theta_i$ goods are preferred by agent $i$, while all remaining goods are nonpreferred. Throughout the proof, $\MMS_i$ is computed with respect to the normalized values $1$ and $\bv_i$. Without loss of generality, we index the agents in nondecreasing order of their thresholds, so that $\theta_1\leq\theta_2\leq\cdots\leq\theta_n$.

For every agent $i$, let $r_i$ denote her \emph{preferred quota}, namely, the number of preferred goods that we aim to allocate to her. Given an integral preferred quota $r\in\{0,\ldots,\theta_i\}$, we define the \emph{demand function} of agent $i$ as
$$
f_i(r):=\left\lceil r+\max\{\MMS_i-r\bv_i,0\}\right\rceil.
$$
Once agent $i$ receives $r$ preferred goods, these goods contribute exactly $r\bv_i$ to her utility. It then suffices to allocate a total of $f_i(r)$ goods to her, because each of the additional $f_i(r)-r$ goods contributes at least $1$.

We refer to any vector $\vr=(r_1,\ldots,r_n)\in\prod_{i=1}^{n}\{0,\ldots,\theta_i\}$ as a \emph{preferred-quota vector}. Given such a vector, an allocation $A$ satisfies the \emph{bivalued cardinality requirement} induced by $\vr$ if, for every agent $i\in[n]$,
$$
|A(i)\cap[\theta_i]|\geq r_i
\qquad\text{and}\qquad
|A(i)|\geq f_i(r_i).
$$
Any allocation satisfying these conditions is an exact MMS allocation. Indeed, for every agent $i$,
$$
v_i(A(i))
\geq r_i\bv_i+f_i(r_i)-r_i
\geq r_i\bv_i+\max\{\MMS_i-r_i\bv_i,0\}
\geq\MMS_i.
$$
We call a preferred-quota vector $\vr$ \emph{implementable} if there exists an allocation satisfying the bivalued cardinality requirement induced by $\vr$. Our objective therefore reduces to finding an implementable preferred-quota vector.

The following lemma gives sufficient conditions for implementability.

\begin{lemma}
\label{lem:implementable_condition}
Let $\vr=(r_1,\ldots,r_n)\in\prod_{i=1}^{n}\{0,\ldots,\theta_i\}$ be a preferred-quota vector. It is implementable if both of the following conditions hold:
\begin{enumerate}
    \item \emph{Prefix capacity:} $\sum_{i=1}^{k}r_i\leq\theta_k$ for every $k\in[n]$;
    \item \emph{Total demand:} $\sum_{i=1}^{n}f_i(r_i)\leq m$.
\end{enumerate}
\end{lemma}

\begin{proof}
For every agent $i\in[n]$, create $r_i$ \emph{preferred slots} and $f_i(r_i)-r_i$ \emph{unrestricted slots}. We construct a bipartite graph whose left-hand vertices are these slots and whose right-hand vertices are the goods in $[m]$. Each preferred slot of agent $i$ is adjacent to every good in $[\theta_i]$, whereas each unrestricted slot is adjacent to every good in $[m]$.

We verify Hall's condition. Let $X$ be a nonempty subset of slots. If $X$ contains an unrestricted slot, then its neighborhood is $[m]$. By the total-demand condition,
$$
|X|\leq\sum_{i=1}^{n}f_i(r_i)\leq m=|N(X)|.
$$
Otherwise, $X$ contains only preferred slots. Let $k$ be the largest index of an agent whose preferred slot belongs to $X$. Since $\theta_1\leq\theta_2\leq\cdots\leq\theta_n$, the neighborhood of $X$ is $[\theta_k]$. By the prefix-capacity condition,
$$
|X|\leq\sum_{i=1}^{k}r_i\leq\theta_k=|N(X)|.
$$
Thus, Hall's condition holds for every subset of slots, and Hall's theorem guarantees a matching that covers all slots.

For each matched slot, allocate its matched good to the associated agent. The matched goods give agent $i$ exactly $f_i(r_i)$ goods, including the $r_i$ goods assigned to her preferred slots, all of which lie in $[\theta_i]$. Allocate any unmatched goods arbitrarily; doing so can only increase both cardinality requirements. The resulting allocation satisfies the bivalued cardinality requirement induced by $\vr$, so $\vr$ is implementable.
\end{proof}

\paragraph{The best of $n$ worlds.}
The key idea for finding an implementable preferred-quota vector is to construct $n$ candidate vectors $\vr^{(1)},\ldots,\vr^{(n)}$, which we view as $n$ different \emph{worlds}. In world $c\in[n]$, the entry $r_i^{(c)}$ denotes the preferred quota of agent $i$. We collect these entries into an $n\times n$ preferred-quota table $\mathcal{R}=(r_i^{(c)})_{i,c\in[n]}$. Throughout, rows index agents and columns index worlds: row $i$ records agent $i$'s quotas across all worlds, while column $c$ is the preferred-quota vector $\vr^{(c)}$ of world $c$.

We maintain two invariants across the $n$ worlds. First, every world satisfies the prefix-capacity condition:
$$
\sum_{i=1}^{k}r_i^{(c)}\leq\theta_k
\qquad
\text{for every }c,k\in[n].
$$
Second, the aggregate demand across all worlds satisfies
$$
\sum_{c=1}^{n}\sum_{i=1}^{n}f_i(r_i^{(c)})\leq nm+n-1.
$$
Since every $f_i(r_i^{(c)})$ is integral, there must exist a world $c^\star\in[n]$ such that
$$
\sum_{i=1}^{n}f_i(r_i^{(c^\star)})\leq m.
$$
Indeed, if every world had demand at least $m+1$, then the aggregate demand would be at least $nm+n$, contradicting the second invariant. Therefore, $\vr^{(c^\star)}$ satisfies both conditions in \Cref{lem:implementable_condition}. It is consequently implementable and yields an exact MMS allocation.

\paragraph{Fractional construction and prefix-preserving rounding.}
It remains to construct an integral $n$-world preferred-quota table satisfying the two invariants. In \Cref{sec:fractional-demand}, we replace the original demand function $f_i$ with its lower convex envelope $g_i$, which satisfies $g_i(r)\leq f_i(r)<g_i(r)+1$ at every integral quota $r$. This relaxation admits a fractional table that satisfies the prefix-capacity condition in every world and has total relaxed demand at most $nm$.

In \Cref{sec:extreme-point-solution}, we preserve each agent's total relaxed demand while moving to an extreme point of a weighted transportation polytope. Its non-bound entries form a forest, whereas every bound entry is an integral quota at which $f_i$ and $g_i$ coincide. Finally, \Cref{sec:prefix-rounding} uses a max-flow argument to round the non-bound entries while preserving every prefix-capacity condition. Each such entry incurs an additive gap of less than one, and the forest structure bounds the total gap by $n-1$. Consequently, the resulting integral table $\hat r_i^{(c)}$ satisfies
$$
\sum_{c=1}^{n}\sum_{i=1}^{n}f_i(\hat r_i^{(c)})\leq nm+n-1.
$$
Together with the best-of-$n$-worlds argument, this yields an implementable preferred-quota vector and hence an exact MMS allocation.

\section{The Lower Convex Envelope of Demand}
\label{sec:fractional-demand}

In this section, we introduce a relaxed demand function for fractional preferred quotas. Recall that, for every agent $i\in[n]$ and every integral preferred quota $r\in\{0,\ldots,\theta_i\}$, the original demand is
$$
f_i(r)=\left\lceil r+\max\{\MMS_i-r\bv_i,0\}\right\rceil.
$$
We begin with the uniform fractional table defined by $r_i^{(c)}=\theta_i/n$ for every agent $i\in[n]$ and every world $c\in[n]$. This table satisfies the prefix-capacity condition in every world. Indeed, for every $c,k\in[n]$,
$$
\sum_{i=1}^{k}r_i^{(c)}=\frac{1}{n}\sum_{i=1}^{k}\theta_i\leq\frac{k}{n}\theta_k\leq\theta_k,
$$
where the first inequality follows from $\theta_1\leq\cdots\leq\theta_n$.

It remains to define a suitable demand function for fractional quotas. Directly extending the ceiling-based formula is problematic because the uniform table may violate the desired aggregate demand bound. Indeed, the ceiling is applied independently to each agent--world pair and can introduce an additive overhead arbitrarily close to one at each of the $n^2$ entries. Thus, this direct extension may violate the desired bound even before rounding.

A first attempt is to remove the ceiling and consider the fractional demand function
$$
\ell_i(r):=r+\max\{\MMS_i-r\bv_i,0\}.
$$
This relaxation is sufficient for bounding the uniform table. Since the average value of all goods upper-bounds the MMS benchmark, we have
$$
\MMS_i\leq\frac{\theta_i\bv_i+m-\theta_i}{n}.
$$
It follows that
$$
n\ell_i\left(\frac{\theta_i}{n}\right)=\theta_i+\max\{n\MMS_i-\theta_i\bv_i,0\}\leq m.
$$
Thus, under $\ell_i$, each agent's total fractional demand across the $n$ worlds is at most $m$. Summing over all agents gives a bound of $nm$.

The limitation of $\ell_i$ becomes apparent during rounding, when we must round the quotas to integers and reverse the relaxation by replacing $\ell_i$ with the original demand function $f_i$. For every integral quota $r$, we have
$$
f_i(r)=\lceil\ell_i(r)\rceil
\qquad\text{and hence}\qquad
0\leq f_i(r)-\ell_i(r)<1~.
$$
Thus, each entry incurs less than one unit of additional demand. However, this gap may be positive at many entries and can therefore accumulate across many agent--world pairs.

We thus seek a convex relaxation that is more suitable for recovering an integral solution. In particular, we would like the relaxation to coincide with $f_i$ at more integral quotas, since every such quota incurs no additive gap when we pass back to the original demand. This motivates the lower convex envelope of the integral demand points.

For each agent $i\in[n]$, consider the points $\{(r,f_i(r)):r\in\{0,\ldots,\theta_i\}\}$. We take their lower convex hull and connect consecutive vertices on its lower boundary by line segments. For every $r\in[0,\theta_i]$, let $g_i(r)$ denote the height of the resulting polygonal curve at $r$. We refer to $g_i$ as the \emph{relaxed demand function} of agent $i$. Equivalently, $g_i$ is the lower convex envelope of the integral demand points. We call an integral quota $r$ \emph{tight} if $g_i(r)=f_i(r)$.

The following lemma records two basic properties of the lower convex envelope. The first bounds the gap between the relaxed and original demand functions at every integral quota, while the second characterizes the quotas at which this gap is zero.

\begin{lemma}
\label{lem:lower_convex_envelope}
For every agent $i\in[n]$ and every integral quota $r\in\{0,\ldots,\theta_i\}$, the following hold:
\begin{enumerate}[label=\textup{(\roman*)}]
    \item $g_i(r)\leq f_i(r)<g_i(r)+1$;
    \item $g_i(r)=f_i(r)$ if and only if $(r,f_i(r))$ lies on the lower boundary of the convex hull. In particular, the endpoints of every maximal linear segment of $g_i$ are integral tight quotas.
\end{enumerate}
\end{lemma}

\begin{proof}
Fix an agent $i\in[n]$. By definition, the graph of $g_i$ is the lower boundary of the convex hull of the integral demand points. Therefore, for every integral quota $r$, this boundary lies at or below $(r,f_i(r))$, and hence $g_i(r)\leq f_i(r)$.

To prove the strict upper bound, recall that $\ell_i$ is convex and satisfies $\ell_i(r)\leq f_i(r)$ at every integral quota. Since $g_i$ is the largest convex function lying below the integral demand points, we have $\ell_i(r)\leq g_i(r)$ for every $r\in[0,\theta_i]$. Consequently, for every integral quota $r$,
$$
f_i(r)=\lceil\ell_i(r)\rceil<\ell_i(r)+1\leq g_i(r)+1.
$$
This proves the first property.

For the second property, the geometric definition of $g_i$ implies that $g_i(r)=f_i(r)$ precisely when $(r,f_i(r))$ lies on the lower boundary of the convex hull. Otherwise, the point lies strictly above this boundary, and hence $g_i(r)<f_i(r)$. Every vertex of the lower hull is one of the original demand points and therefore has an integral first coordinate.
\end{proof}

The two properties in \Cref{lem:lower_convex_envelope} play complementary roles in the rounding analysis. The first guarantees that replacing $g_i$ with $f_i$ at any integral quota incurs an additive gap strictly smaller than one. The second identifies the tight quotas, at which this gap is zero. In particular, the endpoints of every maximal linear segment of $g_i$ incur no gap.

The geometric definition of $g_i$ also admits a useful mixing interpretation. A fractional quota $r$ can be viewed as a mixture of integral quotas. If an integral quota $r'$ is assigned weight $\lambda_{r'}$, then $\sum_{r'}\lambda_{r'}r'$ is the average quota and $\sum_{r'}\lambda_{r'}f_i(r')$ is the corresponding average original demand. Equivalently,
$$
g_i(r)=\min\left\{\sum_{r'=0}^{\theta_i}\lambda_{r'}f_i(r'):\sum_{r'=0}^{\theta_i}\lambda_{r'}r'=r,\ \sum_{r'=0}^{\theta_i}\lambda_{r'}=1,\ \lambda_{r'}\geq0\text{ for every }r'\right\}.
$$
Therefore, $g_i(r)$ is the minimum average original demand achievable by mixing integral quotas whose average is $r$. In this sense, $g_i$ does not simply remove the ceiling; it is the tightest convex relaxation induced by the original demands at integral quotas. It is precisely this mixing interpretation that allows us to use the bundle quotas in an MMS partition, which have average $\theta_i/n$, to establish the following per-agent bound for the uniform fractional table.

\begin{lemma}
\label{lem:uniform_relaxed_demand}
For every agent $i\in[n]$, the uniform fractional table satisfies
$$
\sum_{c=1}^{n}g_i\left(\frac{\theta_i}{n}\right)=n g_i\left(\frac{\theta_i}{n}\right)\leq m.
$$
\end{lemma}

\begin{proof}
Fix an agent $i\in[n]$ and an MMS partition $(B_i^{(1)},\ldots,B_i^{(n)})$ such that $v_i(B_i^{(c)})\geq\MMS_i$ for every $c\in[n]$. Let $\bar r_i^{(c)}:=|B_i^{(c)}\cap[\theta_i]|$ be the number of preferred goods in $B_i^{(c)}$. Since the bundles form a partition of $[m]$, every preferred good appears in exactly one bundle, and hence
$$
\sum_{c=1}^{n}\bar r_i^{(c)}=\theta_i.
$$
For every $c\in[n]$, we have
$$
v_i(B_i^{(c)})=\bar r_i^{(c)}\bv_i+\bigl(|B_i^{(c)}|-\bar r_i^{(c)}\bigr)\geq\MMS_i.
$$
The quantity $|B_i^{(c)}|-\bar r_i^{(c)}$ is a nonnegative integer. It follows that
$$
|B_i^{(c)}|\geq\left\lceil\bar r_i^{(c)}+\max\{\MMS_i-\bar r_i^{(c)}\bv_i,0\}\right\rceil=f_i(\bar r_i^{(c)}).
$$
Summing over all bundles gives
$$
\sum_{c=1}^{n}f_i(\bar r_i^{(c)})\leq\sum_{c=1}^{n}|B_i^{(c)}|=m.
$$
The quotas $\bar r_i^{(1)},\ldots,\bar r_i^{(n)}$, each assigned weight $1/n$, form a mixture of integral quotas with average $\theta_i/n$. By the mixing interpretation of $g_i$,
$$
g_i\left(\frac{\theta_i}{n}\right)\leq\frac{1}{n}\sum_{c=1}^{n}f_i(\bar r_i^{(c)})\leq\frac{m}{n}.
$$
Therefore, $n g_i(\theta_i/n)\leq m$, as desired.
\end{proof}

We next preserve the per-agent bound in \Cref{lem:uniform_relaxed_demand} while moving to a sparse fractional quota table. Its sparsity will control the additive gap incurred by integral rounding.

\section{A Fractional Solution for Low-Loss Rounding}
\label{sec:extreme-point-solution}

In this section, we construct a fractional preferred-quota table suitable for rounding. The table satisfies the prefix-capacity condition in every world and preserves each agent's total relaxed demand from the uniform table. We choose an extreme point of a weighted transportation polytope so that only a sparse collection of entries lies strictly between its bounds; only these entries can incur a gap when $g_i$ is replaced by $f_i$ after rounding.

\paragraph{Restricting each row to one affine segment.}
Recall that $g_i$ is convex and piecewise linear. If $\theta_i=0$, set $\alpha_i=\beta_i=0$. If $\theta_i/n$ is a vertex of the graph of $g_i$, including an endpoint of its domain, set $\alpha_i=\beta_i=\theta_i/n$. Otherwise, let $[\alpha_i,\beta_i]$ be the domain of the unique maximal linear segment whose relative interior contains $\theta_i/n$. By \Cref{lem:lower_convex_envelope}, both endpoints are integral tight quotas; that is,
$$
g_i(\alpha_i)=f_i(\alpha_i)
\qquad\text{and}\qquad
g_i(\beta_i)=f_i(\beta_i).
$$
The reason for restricting the quotas of agent $i$ to this segment is that $g_i$ is affine on $[\alpha_i,\beta_i]$. We may therefore redistribute the preferred quotas of agent $i$ among the $n$ worlds without changing the sum of their $g_i$-values, provided that every quota remains in this interval and their sum remains $\theta_i$. More precisely, if $\alpha_i\leq r_i^{(c)}\leq\beta_i$ for every $c\in[n]$ and $\sum_{c=1}^{n}r_i^{(c)}=\theta_i$, then
$$
\sum_{c=1}^{n}g_i(r_i^{(c)})=n g_i\left(\frac{\theta_i}{n}\right).
$$
Thus, the line-segment restriction allows us to move away from the uniform table while preserving each agent's total relaxed demand across all worlds.

\paragraph{A baseline feasible polytope.}
We first consider the polytope containing all fractional preferred-quota tables that satisfy per-agent quota conservation, per-world prefix capacity, and the line-segment restrictions:
$$
\mathcal{P}:=
\left\{
(r_i^{(c)})_{i,c\in[n]}
\;\middle|\;
\begin{aligned}
\sum_{c=1}^{n}r_i^{(c)} &= \theta_i
&& \text{for every }i\in[n],\\
\sum_{i=1}^{k}r_i^{(c)} &\leq \theta_k
&& \text{for every }c,k\in[n],\\
\alpha_i\leq r_i^{(c)} &\leq \beta_i
&& \text{for every }i,c\in[n]
\end{aligned}
\right\}.
$$
The uniform fractional table belongs to $\mathcal{P}$, so it is nonempty. More importantly, every table in $\mathcal{P}$ preserves the per-agent bound obtained from the uniform table under $g_i$.

\begin{lemma}[Preservation under the relaxed demand function]
\label{lem:preserved_relaxed_demand}
For every table $(r_i^{(c)})_{i,c\in[n]}\in\mathcal{P}$ and every agent $i\in[n]$,
$$
\sum_{c=1}^{n}g_i(r_i^{(c)})=n g_i\left(\frac{\theta_i}{n}\right)\leq m.
$$
\end{lemma}

\begin{proof}
Fix an agent $i\in[n]$. By the definition of $\mathcal{P}$, every $r_i^{(c)}$ lies in $[\alpha_i,\beta_i]$, on which $g_i$ is affine. Since $\sum_{c=1}^{n}r_i^{(c)}=\theta_i$, we obtain
$$
\sum_{c=1}^{n}g_i(r_i^{(c)})
=
n g_i\left(\frac{1}{n}\sum_{c=1}^{n}r_i^{(c)}\right)
=
n g_i\left(\frac{\theta_i}{n}\right).
$$
The inequality follows from \Cref{lem:uniform_relaxed_demand}.
\end{proof}

Although every point in $\mathcal{P}$ satisfies the required prefix-capacity conditions and per-agent bounds under $g_i$, not every point is equally suitable for rounding. We call an entry $r_i^{(c)}$ a \emph{bound entry} if $r_i^{(c)}\in\{\alpha_i,\beta_i\}$ and a \emph{non-bound entry} otherwise. Since $\alpha_i$ and $\beta_i$ are integral tight quotas, every bound entry is integral and satisfies $f_i(r_i^{(c)})=g_i(r_i^{(c)})$. Thus, bound entries neither change during rounding nor incur a gap when we pass back from $g_i$ to $f_i$. By contrast, a non-bound entry need not be a tight quota and may therefore incur a positive gap when we pass back to $f_i$, even if the entry itself is integral. Our goal is therefore to find a point in $\mathcal{P}$ whose non-bound entries form a sparse structure.

\paragraph{A weighted transportation polytope.}
The polytope $\mathcal{P}$ alone does not provide the desired sparsity. We therefore replace its prefix-capacity constraints with a stronger per-world weighted-balance condition that gives the polytope a transportation structure. This structure allows us to exploit the fact that an extreme point has only a sparse collection of non-bound entries. Define
$$
\lambda:=\frac{1}{n}\sum_{\substack{i\in[n]\\\theta_i>0}}\frac{\theta_i-n\alpha_i}{\theta_i},
$$
and consider the following smaller polytope:
$$
\mathcal{Q}:=
\left\{
(r_i^{(c)})_{i,c\in[n]}
\;\middle|\;
\begin{aligned}
\sum_{c=1}^{n}r_i^{(c)} &= \theta_i
&& \text{for every }i\in[n],\\
\sum_{\substack{i\in[n]\\\theta_i>0}}\frac{r_i^{(c)}-\alpha_i}{\theta_i} &= \lambda
&& \text{for every }c\in[n],\\
\alpha_i\leq r_i^{(c)} &\leq \beta_i
&& \text{for every }i,c\in[n]
\end{aligned}
\right\}.
$$
We call $\mathcal{Q}$ a \emph{weighted transportation polytope}. The first family of constraints fixes the total quota of each agent, while the second balances the normalized residual quotas in every world. For every agent with $\theta_i>0$, shift each entry by $\alpha_i$ and scale it by $1/\theta_i$. After omitting the rows with $\theta_i=0$, which are fixed at zero, the resulting variables have fixed row and column sums together with lower and upper bounds. Thus, $\mathcal{Q}$ is affinely equivalent to a capacitated transportation polytope.

The uniform fractional table belongs to $\mathcal{Q}$ because, for every world $c\in[n]$,
$$
\sum_{\substack{i\in[n]\\\theta_i>0}}\frac{\theta_i/n-\alpha_i}{\theta_i}=\lambda.
$$
Hence, $\mathcal{Q}$ is nonempty. The following lemma shows that the per-world weighted-balance condition implies every prefix-capacity condition and therefore that $\mathcal{Q}\subseteq\mathcal{P}$.

\begin{lemma}[Weighted balance implies prefix capacity]
\label{lem:weighted_balance_prefix}
Every table $(r_i^{(c)})_{i,c\in[n]}\in\mathcal{Q}$ belongs to $\mathcal{P}$.
\end{lemma}

\begin{proof}
The per-agent quota-conservation and line-segment constraints of $\mathcal{P}$ follow directly from the definition of $\mathcal{Q}$. It remains to verify the prefix-capacity conditions.

Fix a world $c\in[n]$ and a prefix $k\in[n]$. If $\theta_k=0$, then $\theta_i=0$ for every $i\leq k$. Consequently, $\alpha_i=\beta_i=r_i^{(c)}=0$ for every $i\leq k$, and the prefix-capacity condition holds immediately. Suppose that $\theta_k>0$. Since $\theta_i\leq\theta_k$ for every $i\leq k$ and $r_i^{(c)}-\alpha_i\geq0$, the per-world weighted-balance condition gives
$$
\sum_{i=1}^{k}\bigl(r_i^{(c)}-\alpha_i\bigr)
=
\sum_{\substack{i\leq k\\\theta_i>0}}
\theta_i\frac{r_i^{(c)}-\alpha_i}{\theta_i}
\leq
\theta_k
\sum_{\substack{i\leq k\\\theta_i>0}}
\frac{r_i^{(c)}-\alpha_i}{\theta_i}
\leq
\theta_k\lambda.
$$
Moreover, by the definition of $\lambda$,
$$
\lambda
=
\frac{|\{i\in[n]:\theta_i>0\}|}{n}
-
\sum_{\substack{i\in[n]\\\theta_i>0}}\frac{\alpha_i}{\theta_i}
\leq
1-
\sum_{\substack{i\leq k\\\theta_i>0}}\frac{\alpha_i}{\theta_i}
\leq
1-\frac{1}{\theta_k}\sum_{i=1}^{k}\alpha_i.
$$
The first inequality follows because $|\{i\in[n]:\theta_i>0\}|/n\leq1$ and the omitted terms are nonnegative. The second follows from $\theta_i\leq\theta_k$ for every $i\leq k$. Therefore,
$$
\sum_{i=1}^{k}r_i^{(c)}
=
\sum_{i=1}^{k}\alpha_i
+
\sum_{i=1}^{k}\bigl(r_i^{(c)}-\alpha_i\bigr)
\leq
\sum_{i=1}^{k}\alpha_i+\theta_k\lambda
\leq
\theta_k.
$$
Thus, every prefix-capacity condition is satisfied, and hence $\mathcal{Q}\subseteq\mathcal{P}$.
\end{proof}

\paragraph{Sparsity at an extreme point.}
Since $\mathcal{Q}$ is a nonempty bounded polytope, it has an extreme point. Fix any extreme point
$$
(\tilde r_i^{(c)})_{i,c\in[n]}\in\mathcal{Q}.
$$
By \Cref{lem:weighted_balance_prefix}, this extreme point belongs to $\mathcal{P}$ and therefore satisfies every prefix-capacity condition. By \Cref{lem:preserved_relaxed_demand}, it also satisfies
$$
\sum_{c=1}^{n}g_i(\tilde r_i^{(c)})
=
n g_i\left(\frac{\theta_i}{n}\right)
\leq m
\qquad
\text{for every }i\in[n].
$$
We use this extreme point as the fractional preferred-quota table to be rounded.

The advantage of choosing an extreme point is that its non-bound entries have a sparse structure. Construct a bipartite graph whose left-hand vertices are the agents, whose right-hand vertices are the worlds, and whose edges correspond precisely to the non-bound entries; that is, $(i,c)$ is an edge if $\alpha_i<\tilde r_i^{(c)}<\beta_i$.

The next lemma is the bounded-variable counterpart of the classical forest characterization of transportation bases \citep{dantzig1951application}. In the classical setting, a cycle of positive entries admits an alternating perturbation. Here, the same argument applies to entries strictly between their bounds, after scaling the perturbation in row $i$ by $\theta_i$.

\begin{lemma}[Extreme-point forest]
\label{lem:extreme_point_forest}
The bipartite graph of non-bound entries is a forest.
\end{lemma}

\begin{proof}
Suppose, for contradiction, that the graph contains a cycle. Assign alternating signs to the edges of this cycle. For an edge $(i,c)$ with a positive sign, increase $\tilde r_i^{(c)}$ by $\theta_i\varepsilon$; for an edge with a negative sign, decrease it by $\theta_i\varepsilon$.

Every agent on the cycle is incident to one positive and one negative cycle edge. Hence, the two perturbations associated with that agent cancel, and its total quota across all worlds remains unchanged. Similarly, every world on the cycle is incident to one positive and one negative cycle edge. Since every agent appearing on the cycle satisfies $\theta_i>0$, each normalized perturbation has magnitude
$$
\frac{\theta_i\varepsilon}{\theta_i}=\varepsilon.
$$
The two contributions to the per-world weighted-balance constraint therefore cancel.

Because every entry on the cycle lies strictly between its lower and upper bounds, both directions of the perturbation remain feasible for sufficiently small $\varepsilon>0$. The extreme point is therefore the midpoint of two distinct feasible points in $\mathcal{Q}$, contradicting its extremality. Hence, the graph contains no cycle and is a forest.
\end{proof}

The forest structure shows that the non-bound entries of the extreme point are sparse, while every remaining entry is equal to either $\alpha_i$ or $\beta_i$. All bound entries are integral tight quotas and therefore require no rounding and introduce no gap between $g_i$ and $f_i$.

In the next section, we round $(\tilde r_i^{(c)})_{i,c\in[n]}$ to an integral preferred-quota table while preserving every prefix-capacity condition. The per-agent bounds
$$
\sum_{c=1}^{n}g_i(\tilde r_i^{(c)})\leq m
$$
and the forest structure will together imply that the original aggregate demand of the resulting integral table is at most $nm+n-1$.

\section{Rounding While Preserving Prefix Capacity}
\label{sec:prefix-rounding}

In the final step, we round the extreme-point fractional table $(\tilde r_i^{(c)})_{i,c\in[n]}$ constructed in the previous section to an integral table while preserving every prefix-capacity condition. We then use the forest structure of its non-bound entries to show that passing from the relaxed demand functions $g_i$ back to the original demand functions $f_i$ incurs a total additive loss of at most $n-1$.

For every entry $\tilde r_i^{(c)}$, we round it to either its floor or its ceiling; that is, we require $\hat r_i^{(c)}\in\{\lfloor\tilde r_i^{(c)}\rfloor,\lceil\tilde r_i^{(c)}\rceil\}$. Rounding the entries independently may violate a prefix-capacity condition. We therefore coordinate the rounding so that every rounded prefix sum is at most the ceiling of its fractional value. The existence of such a rounding follows from a circulation argument.

\begin{lemma}
\label{lem:prefix_preserving_rounding}
The table $(\tilde r_i^{(c)})_{i,c\in[n]}$ admits an integral rounding $(\hat r_i^{(c)})_{i,c\in[n]}$ satisfying the following properties:
\begin{enumerate}
    \item $\hat r_i^{(c)}\in\{\lfloor\tilde r_i^{(c)}\rfloor,\lceil\tilde r_i^{(c)}\rceil\}$ for every $i,c\in[n]$;
    \item $\sum_{c=1}^{n}\hat r_i^{(c)}=\theta_i$ for every $i\in[n]$;
    \item $\sum_{i=1}^{k}\hat r_i^{(c)}\leq\left\lceil\sum_{i=1}^{k}\tilde r_i^{(c)}\right\rceil\leq\theta_k$ for every $c,k\in[n]$.
\end{enumerate}
Moreover, such an integral table can be computed in polynomial time using a max-flow algorithm.
\end{lemma}

\begin{proof}
We construct a flow network with lower and upper bounds. The network contains a source vertex $s$, a vertex $u_i$ for every agent $i\in[n]$, vertices $v_{i,c}$ for every $i,c\in[n]$, and a sink vertex $t$. For every agent $i\in[n]$, we add an edge from $s$ to $u_i$ whose lower and upper bounds are both $\theta_i$. For every $i,c\in[n]$, we add an edge from $u_i$ to $v_{i,c}$ with lower bound $\lfloor\tilde r_i^{(c)}\rfloor$ and upper bound $\lceil\tilde r_i^{(c)}\rceil$.

For every world $c\in[n]$ and every $k\in[n-1]$, we add a forward chain edge from $v_{k,c}$ to $v_{k+1,c}$ with lower bound $0$ and upper bound $\left\lceil\sum_{i=1}^{k}\tilde r_i^{(c)}\right\rceil$. We complete this chain with an edge from $v_{n,c}$ to $t$ having lower bound $0$ and upper bound $\left\lceil\sum_{i=1}^{n}\tilde r_i^{(c)}\right\rceil$. Finally, we add an edge from $t$ to $s$ whose lower and upper bounds are both $\sum_{i=1}^{n}\theta_i$, thereby turning the network into a circulation network.

The fractional table $(\tilde r_i^{(c)})_{i,c\in[n]}$ induces a feasible fractional circulation. Specifically, we send $\tilde r_i^{(c)}$ units of flow from $u_i$ to $v_{i,c}$ and $\sum_{i'=1}^{k}\tilde r_{i'}^{(c)}$ units through the forward edge leaving $v_{k,c}$, where for $k=n$ this is the edge from $v_{n,c}$ to $t$. The row-sum constraints defining $\mathcal Q$ guarantee flow conservation at every $u_i$, and the definition of the chain flows guarantees conservation at every $v_{k,c}$. All edge bounds are satisfied by construction.

Since all lower and upper bounds are integral, the integrality theorem for feasible circulations guarantees an integral feasible circulation. Let $\hat r_i^{(c)}$ be the integral flow on the edge from $u_i$ to $v_{i,c}$. Its lower and upper bounds imply that
$$
\hat r_i^{(c)}
\in
\left\{
\lfloor\tilde r_i^{(c)}\rfloor,
\lceil\tilde r_i^{(c)}\rceil
\right\}.
$$
Flow conservation at $u_i$ gives $\sum_{c=1}^{n}\hat r_i^{(c)}=\theta_i$. Moreover, flow conservation along the chain associated with world $c$ implies that the forward edge leaving $v_{k,c}$ carries exactly $\sum_{i=1}^{k}\hat r_i^{(c)}$ units of flow. Its upper bound therefore gives
$$
\sum_{i=1}^{k}\hat r_i^{(c)}
\leq
\left\lceil\sum_{i=1}^{k}\tilde r_i^{(c)}\right\rceil.
$$
Since $(\tilde r_i^{(c)})_{i,c\in[n]}\in\mathcal P$, we have $\sum_{i=1}^{k}\tilde r_i^{(c)}\leq\theta_k$. As $\theta_k$ is integral,
$$
\left\lceil\sum_{i=1}^{k}\tilde r_i^{(c)}\right\rceil\leq\theta_k.
$$
Thus, the integral table satisfies all the desired properties.
\end{proof}

It remains to bound the increase incurred when we pass from the relaxed demand functions $g_i$ back to the original demand functions $f_i$. Recall that an entry $\tilde r_i^{(c)}$ is a \emph{bound entry} if $\tilde r_i^{(c)}\in\{\alpha_i,\beta_i\}$ and a \emph{non-bound entry} otherwise. Since $\alpha_i$ and $\beta_i$ are integral, every bound entry is unchanged by rounding. Both endpoints are tight quotas, so, for every bound entry,
$$
f_i(\hat r_i^{(c)})=f_i(\tilde r_i^{(c)})=g_i(\tilde r_i^{(c)})~.
$$
Hence, only non-bound entries can contribute to the increase. The following lemma bounds their total contribution for each agent.

\begin{lemma}
\label{lem:agent_rounding_loss}
For every agent $i\in[n]$, let $k_i:=\left|\left\{c\in[n]:\alpha_i<\tilde r_i^{(c)}<\beta_i\right\}\right|$ denote the number of non-bound entries in row $i$. Then
$$
\sum_{c=1}^{n}f_i(\hat r_i^{(c)})
\leq
\left\lceil\sum_{c=1}^{n}g_i(\tilde r_i^{(c)})\right\rceil
+
\max\{k_i-1,0\}~.
$$
\end{lemma}

\begin{proof}
Fix an agent $i\in[n]$. Since $\tilde r_i^{(c)}\in[\alpha_i,\beta_i]$ and both endpoints are integral, every rounded quota $\hat r_i^{(c)}$ is an integer in $[\alpha_i,\beta_i]$. By \Cref{lem:lower_convex_envelope}, for every such quota $q$,
$$
f_i(q)<g_i(q)+1~.
$$
Because $g_i$ is affine on $[\alpha_i,\beta_i]$ and the rounding preserves the total quota of agent $i$, we also have
$$
\sum_{c=1}^{n}g_i(\hat r_i^{(c)})
=
\sum_{c=1}^{n}g_i(\tilde r_i^{(c)})~.
$$
If $k_i=0$, every entry is bound, and hence $f_i(\hat r_i^{(c)})=g_i(\tilde r_i^{(c)})$ for every $c\in[n]$. Suppose that $k_i\geq1$. Bound entries contribute no increase, whereas the increase at each of the $k_i$ non-bound entries is strictly smaller than $1$. Therefore,
$$
\sum_{c=1}^{n}f_i(\hat r_i^{(c)})
<
\sum_{c=1}^{n}g_i(\hat r_i^{(c)})+k_i
=
\sum_{c=1}^{n}g_i(\tilde r_i^{(c)})+k_i~.
$$
Since the left-hand side is integral, the strict inequality implies
$$
\sum_{c=1}^{n}f_i(\hat r_i^{(c)})
\leq
\left\lceil\sum_{c=1}^{n}g_i(\tilde r_i^{(c)})\right\rceil+k_i-1~.
$$
Combining the two cases proves the lemma.
\end{proof}

The saving of one unit in the preceding lemma is essential. For every agent with at least one non-bound entry, the strict entrywise gap and integrality improve the naive loss bound from $k_i$ to $k_i-1$. We now combine these per-agent savings with the forest structure to bound the aggregate loss.

\begin{lemma}
\label{lem:aggregate_demand_bound}
The rounded table satisfies
$$
\sum_{c=1}^{n}\sum_{i=1}^{n}f_i(\hat r_i^{(c)})\leq nm+n-1~.
$$
\end{lemma}

\begin{proof}
Let $E=\sum_{i=1}^{n}k_i$ be the total number of non-bound entries, and let $a$ be the number of agents associated with at least one non-bound entry. Delete all isolated vertices from the bipartite graph of non-bound entries, and let $b$ be the number of world vertices that remain. By \Cref{lem:extreme_point_forest}, the resulting graph is a forest with $E$ edges and $a+b$ vertices, where $b\leq n$. If $E>0$, the forest property implies
$$
E\leq a+b-1\leq a+n-1~.
$$
Consequently,
$$
\sum_{i=1}^{n}\max\{k_i-1,0\}=E-a\leq n-1~.
$$
The inequality also holds trivially when $E=0$. Combining it with \Cref{lem:agent_rounding_loss} gives
$$
\sum_{i=1}^{n}\sum_{c=1}^{n}f_i(\hat r_i^{(c)})
\leq
\sum_{i=1}^{n}\left\lceil\sum_{c=1}^{n}g_i(\tilde r_i^{(c)})\right\rceil+n-1~.
$$
By the construction of $(\tilde r_i^{(c)})_{i,c\in[n]}$, for every agent $i\in[n]$,
$$
\sum_{c=1}^{n}g_i(\tilde r_i^{(c)})
=
n g_i\left(\frac{\theta_i}{n}\right)
\leq m~.
$$
Thus, $\left\lceil\sum_c g_i(\tilde r_i^{(c)})\right\rceil\leq m$ for every agent $i$. Substituting these bounds into the preceding inequality proves the lemma.
\end{proof}

We now restate and prove the goods theorem.

\exactMMSGoodsTheorem*

\begin{proof}
By \Cref{lem:zero_lower_value_reduction}, it suffices to consider instances in which both agent-specific values are positive. By the standard IDO reduction cited in the preliminaries, we may further restrict attention to IDO instances. Finally, positive scaling of an agent's valuation scales her MMS by the same factor, so we may normalize her lower value to $1$. It therefore suffices to prove the theorem for the normalized positive IDO instance considered above.

Let $(\tilde r_i^{(c)})_{i,c\in[n]}$ be the extreme-point fractional table constructed in the previous section, and let $(\hat r_i^{(c)})_{i,c\in[n]}$ be the integral rounding given by \Cref{lem:prefix_preserving_rounding}. For every world $c\in[n]$ and every prefix $k\in[n]$, the rounded table satisfies
$$
\sum_{i=1}^{k}\hat r_i^{(c)}\leq\theta_k.
$$
Moreover, by \Cref{lem:aggregate_demand_bound},
$$
\sum_{c=1}^{n}\sum_{i=1}^{n}f_i(\hat r_i^{(c)})\leq nm+n-1~.
$$
Since every $f_i(\hat r_i^{(c)})$ is integral, if every world had total demand at least $m+1$, then
$$
\sum_{c=1}^{n}\sum_{i=1}^{n}f_i(\hat r_i^{(c)})
\geq
n(m+1)
=
nm+n,
$$
contradicting the preceding upper bound. Hence, there exists a world $c^\star\in[n]$ such that
$$
\sum_{i=1}^{n}f_i(\hat r_i^{(c^\star)})\leq m.
$$
Therefore, the preferred-quota vector $(\hat r_1^{(c^\star)},\ldots,\hat r_n^{(c^\star)})$ satisfies both conditions of \Cref{lem:implementable_condition}: for every $k\in[n]$, $\sum_{i=1}^{k}\hat r_i^{(c^\star)}\leq\theta_k$, and its total demand satisfies $\sum_{i=1}^{n}f_i(\hat r_i^{(c^\star)})\leq m$. By \Cref{lem:implementable_condition}, this preferred-quota vector is implementable. Thus, there exists an allocation $A$ such that every agent $i\in[n]$ receives at least $\hat r_i^{(c^\star)}$ preferred goods and at least $f_i(\hat r_i^{(c^\star)})$ goods in total. Since $f_i(r)\geq r+\max\{\MMS_i-r\bv_i,0\}$ for every integral quota $r$, we have
$$
v_i(A(i))
\geq
\hat r_i^{(c^\star)}\bv_i+
f_i(\hat r_i^{(c^\star)})-\hat r_i^{(c^\star)}
\geq
\MMS_i.
$$
Hence, $A$ is an exact MMS allocation for the normalized positive IDO instance. Undoing the normalization and the standard IDO reduction, and then applying \Cref{lem:zero_lower_value_reduction}, yields an exact MMS allocation for the original instance.

It remains to verify the running time. In a bivalued instance, each agent's MMS can be computed in polynomial time by a standard count-based dynamic program over the numbers of preferred and nonpreferred goods. The functions $f_i$ and $g_i$ and the breakpoints of $g_i$ can then be constructed explicitly. A vertex of the rational polytope $\mathcal{Q}$ can be obtained by lexicographically optimizing its $n^2$ coordinates through a polynomial number of linear programs. The rounding circulation and the matching in \Cref{lem:implementable_condition} are computable by polynomial-time flow and matching algorithms, respectively. Finally, the standard IDO reduction and the preprocessing in \Cref{lem:zero_lower_value_reduction} run in polynomial time and preserve polynomial encoding length. Thus, the entire procedure runs in polynomial time.
\end{proof}

\section*{AI Disclosure}
The author used OpenAI's ChatGPT and Codex in both the research and writing of this paper. Long-running autonomous attempts with Codex (GPT-5.6 Sol, extra-high reasoning) did not independently produce the result. The research instead involved more than a week of extensive, many-round discussions between the author and ChatGPT (GPT-5.6 Pro), during which numerous ideas were proposed, tested, and refined. The model made substantive contributions to this process; most notably, it suggested the weighted transportation polytope that became a key component of the proof. AI tools also assisted substantially with writing and checking technical details, especially the symmetric proof for chores. The final proof follows the author's overall framework and ideas. The author independently verified all arguments and takes full responsibility for the content and correctness of the paper.

I would also like to share an interesting story. After learning of the result, a colleague independently launched a separate long-running autonomous Codex session and asked it to solve the problem. The session also used GPT-5.6 Sol, possibly with the reasoning effort set to ultra. The attempt did not succeed. Even after the session was given the proof for goods, it did not successfully reconstruct the proof for chores.

\bibliographystyle{plainnat}
\bibliography{references}

\appendix
\section{Zero Lower Values}
\label{sec:zero-lower-values}

The main proofs normalize each agent's lower value to $1$ and are therefore written for strictly positive item values. We now justify this assumption.

\begin{lemma}
\label{lem:zero_lower_value_reduction}
For both goods and chores, zero lower values can be eliminated in polynomial time, and any exact MMS allocation for the transformed positive instance can be converted into one for the original instance.
\end{lemma}

\begin{proof}
For chores, let $Z$ contain every chore that has zero cost for some agent. Reserve each $j\in Z$ for an agent $i(j)$ with $v_{i(j),j}=0$, and remove $Z$ to obtain an instance $I'$. Removing $Z$ from an MMS partition of agent $i$ shows that
$$
\MMS_i(I')\leq\MMS_i(I)~.
$$
Hence, an exact MMS allocation of $I'$ remains exact for $I$ after each reserved chore $j$ is assigned to $i(j)$. All costs remaining in $I'$ are positive. If an agent realizes fewer than two costs, her unrealized values can be replaced by arbitrary distinct positive rationals without changing the instance.

For goods, the case $m=0$ is immediate. For every agent $i$ with lower value $a_i=0$, replace each zero value by $\epsilon_i:=b_i/(m+1)$. Let $h_i$ be the number of goods worth $b_i$ to agent $i$ and let $q_i:=\lfloor h_i/n\rfloor$. Her original MMS is $q_i b_i$, while her MMS after perturbation is at least $q_i b_i$. If $q_i\geq1$, any bundle containing at most $q_i-1$ high-valued goods has perturbed value at most
$$
(q_i-1)b_i+m\epsilon_i<q_i b_i~.
$$
Thus, every bundle meeting the perturbed MMS contains at least $q_i$ high-valued goods and has original value at least $q_i b_i$. When $q_i=0$, the original MMS constraint is automatic. Therefore, any exact MMS allocation of the perturbed instance is also exact for the original instance. The perturbation has polynomial encoding length, so both reductions run in polynomial time.
\end{proof}

\section{Exact MMS Allocations for Chores}
\label{sec:chores}

We now prove the result for chores. The proof uses the same fractional construction and rounding machinery as for goods, but the objective is reversed: we seek a world whose total \emph{capacity} is at least the number of chores.

By \Cref{lem:zero_lower_value_reduction}, the standard IDO reduction cited in the preliminaries, and positive scaling, it suffices to consider a normalized positive IDO instance. Under the standard IDO convention, rank positions are ordered from most to least costly. We reverse these common rank labels, so that the chores are indexed from easiest to most difficult and, for every agent $i\in[n]$,
$$
v_{ij}=\begin{cases}
1, & j\leq\eta_i,\\
\bv_i, & j>\eta_i,
\end{cases}
\qquad \bv_i>1.
$$
We index the agents so that $\eta_1\leq\cdots\leq\eta_n$.

\paragraph{Integral and relaxed capacities.}
For agent $i$, an easy slot can receive one chore in $[\eta_i]$, whereas an unrestricted slot can receive any one chore. Set
$$
\rho_i:=\min\{\eta_i,\lfloor\MMS_i\rfloor\}
$$
and, for every integral $r\in\{0,\ldots,\rho_i\}$, define
$$
F_i(r):=r+\left\lfloor\frac{\MMS_i-r}{\bv_i}\right\rfloor.
$$
Thus, any assignment that fills the $r$ easy slots with easy chores and the remaining $F_i(r)-r$ unrestricted slots with arbitrary chores has cost at most $r+\bv_i(F_i(r)-r)\leq\MMS_i$. Notice that $F_i(r)-r\geq0$ because $r\leq\rho_i\leq\lfloor\MMS_i\rfloor$.

The domain contains all quotas used below. Indeed, in any MMS partition of agent $i$, the number of easy chores in each bundle is at most $\rho_i$. Since these numbers sum to $\eta_i$, we also have $\eta_i/n\in[0,\rho_i]$.

Let $G_i$ be the least concave majorant of the integral capacity points
$$
\{(r,F_i(r)):r=0,\ldots,\rho_i\}.
$$
We also call $G_i$ their upper concave envelope. An integral quota $r$ is \emph{tight} if $G_i(r)=F_i(r)$.

\begin{lemma}
\label{lem:upper_concave_envelope}
For every agent $i\in[n]$ and every integral $r\in\{0,\ldots,\rho_i\}$,
$$
F_i(r)\leq G_i(r)<F_i(r)+1.
$$
Moreover, the endpoints of every maximal linear segment of $G_i$ are integral tight quotas.
\end{lemma}

\begin{proof}
The first inequality follows from the definition of the upper envelope. The affine function
$$
L_i(r):=r+\frac{\MMS_i-r}{\bv_i}
$$
lies above every integral capacity point and satisfies $F_i(r)=\lfloor L_i(r)\rfloor$ at integral $r$. Since $G_i$ is the smallest concave majorant of these points,
$$
F_i(r)\leq G_i(r)\leq L_i(r)<F_i(r)+1.
$$
The final statement follows from the geometry of the upper hull.
\end{proof}

Equivalently, $G_i(r)$ is the maximum of $\sum_{q=0}^{\rho_i}\lambda_qF_i(q)$ over all vectors $(\lambda_q)_{q=0}^{\rho_i}$ satisfying $\lambda_q\geq0$, $\sum_q\lambda_q=1$, and $\sum_q q\lambda_q=r$.

\begin{lemma}
\label{lem:uniform_relaxed_capacity}
For every agent $i\in[n]$,
$$
nG_i\left(\frac{\eta_i}{n}\right)\geq m.
$$
\end{lemma}

\begin{proof}
Fix an MMS partition $(B_i^{(1)},\ldots,B_i^{(n)})$ for agent $i$, and let $\bar r_i^{(c)}:=|B_i^{(c)}\cap[\eta_i]|$. If $d_i^{(c)}:=|B_i^{(c)}|-\bar r_i^{(c)}$, then
$$
\bar r_i^{(c)}+\bv_i d_i^{(c)}\leq\MMS_i,
$$
so $|B_i^{(c)}|\leq F_i(\bar r_i^{(c)})$. Moreover, $\sum_c\bar r_i^{(c)}=\eta_i$. The mixing interpretation therefore gives
$$
G_i\left(\frac{\eta_i}{n}\right)
\geq\frac1n\sum_{c=1}^nF_i(\bar r_i^{(c)})
\geq\frac mn.
$$
\end{proof}

\paragraph{Fractional construction and rounding.}
If $\rho_i=0$, set $\alpha_i=\beta_i=0$. Otherwise, if $\eta_i/n$ is a vertex of the graph of $G_i$, including an endpoint of its domain, set $\alpha_i=\beta_i=\eta_i/n$; if it is not, let $[\alpha_i,\beta_i]$ be the domain of the unique maximal linear segment whose relative interior contains $\eta_i/n$. In every case, the endpoints are integral tight quotas. Let $\mathcal P_{\mathrm{ch}}$ contain the tables satisfying
$$
\sum_c r_i^{(c)}=\eta_i,\qquad
\alpha_i\leq r_i^{(c)}\leq\beta_i,
\qquad
\sum_{i=1}^k r_i^{(c)}\leq\eta_k
$$
for every $i,c,k\in[n]$. Since $G_i$ is affine on $[\alpha_i,\beta_i]$, every such table satisfies
$$
\sum_{c=1}^nG_i(r_i^{(c)})
=nG_i\left(\frac{\eta_i}{n}\right)\geq m.
$$

Define
$$
\lambda:=\frac1n\sum_{i:\eta_i>0}
\frac{\eta_i-n\alpha_i}{\eta_i}
$$
and let $\mathcal Q_{\mathrm{ch}}$ be the set of tables satisfying the same row-sum and bound constraints together with
$$
\sum_{i:\eta_i>0}\frac{r_i^{(c)}-\alpha_i}{\eta_i}=\lambda
\qquad\text{for every world }c.
$$
The uniform table belongs to $\mathcal Q_{\mathrm{ch}}$.

\begin{lemma}[Chores counterpart of \Cref{lem:weighted_balance_prefix}]
\label{lem:chores_weighted_balance_prefix}
We have $\mathcal Q_{\mathrm{ch}}\subseteq\mathcal P_{\mathrm{ch}}$.
\end{lemma}

\begin{proof}
The proof of \Cref{lem:weighted_balance_prefix} applies verbatim after replacing $\theta_i$, $\mathcal Q$, and $\mathcal P$ by $\eta_i$, $\mathcal Q_{\mathrm{ch}}$, and $\mathcal P_{\mathrm{ch}}$, respectively. Its required side conditions hold because $\eta_1\leq\cdots\leq\eta_n$ and $\alpha_i=0$ whenever $\eta_i=0$.
\end{proof}

\begin{lemma}[Chores counterpart of \Cref{lem:extreme_point_forest}]
\label{lem:chores_extreme_point_forest}
At every extreme point of $\mathcal Q_{\mathrm{ch}}$, the entries strictly between their bounds form a forest on the agents and worlds.
\end{lemma}

\begin{proof}
For completeness, the cycle-perturbation proof of \Cref{lem:extreme_point_forest} applies verbatim, using perturbations $\pm\eta_i\varepsilon$ on edges incident to agent $i$. Every agent on a cycle has $\eta_i>0$, so these perturbations are well defined.
\end{proof}

Fix an extreme point $(\tilde r_i^{(c)})$ of $\mathcal Q_{\mathrm{ch}}$. By the two preceding lemmas, it belongs to $\mathcal P_{\mathrm{ch}}$ and its non-bound entries form a forest.

The same circulation construction as in \Cref{lem:prefix_preserving_rounding} yields the required rounding. For completeness, take edges from $s$ to $u_i$ with fixed flow $\eta_i$ and edges from $u_i$ to $v_{i,c}$ with lower and upper bounds $\lfloor\tilde r_i^{(c)}\rfloor$ and $\lceil\tilde r_i^{(c)}\rceil$. For each world $c$, form a chain $v_{1,c}\to\cdots\to v_{n,c}\to t$ in which the edge leaving $v_{k,c}$ has lower bound $0$ and upper bound $\left\lceil\sum_{i=1}^{k}\tilde r_i^{(c)}\right\rceil$, and close the circulation by an edge from $t$ to $s$ with fixed flow $\sum_i\eta_i$. Sending $\tilde r_i^{(c)}$ units from $u_i$ to $v_{i,c}$ and the corresponding prefix sum along each chain gives a feasible fractional circulation. All bounds are integral, so there is an integral feasible circulation. Let $\hat r_i^{(c)}$ be its flow from $u_i$ to $v_{i,c}$. Flow conservation and the chain capacities give
$$
\hat r_i^{(c)}\in
\{\lfloor\tilde r_i^{(c)}\rfloor,\lceil\tilde r_i^{(c)}\rceil\},
\qquad \sum_c\hat r_i^{(c)}=\eta_i,
$$
and
$$
\sum_{i=1}^k\hat r_i^{(c)}
\leq\left\lceil\sum_{i=1}^k\tilde r_i^{(c)}\right\rceil
\leq\eta_k
$$
for every $c,k\in[n]$.

It remains to bound the capacity lost when $G_i$ is replaced by $F_i$ after rounding. Let $k_i$ be the number of non-bound entries of agent $i$. Bound entries are integral and tight. Because $\alpha_i$ and $\beta_i$ are integral, the floor and ceiling of every point in $[\alpha_i,\beta_i]$ remain in this interval. Hence, every rounded non-bound entry $q$ lies in the domain of $F_i$ and satisfies $F_i(q)>G_i(q)-1$. Affinity and row-sum preservation give
$$
\sum_cG_i(\hat r_i^{(c)})=\sum_cG_i(\tilde r_i^{(c)}).
$$
If $k_i=0$, every entry is bound, so
$$
\sum_cF_i(\hat r_i^{(c)})
=\sum_cG_i(\tilde r_i^{(c)}).
$$
If $k_i\geq1$, summing the strict gap over the non-bound entries gives
$$
\sum_cF_i(\hat r_i^{(c)})
>
\sum_cG_i(\tilde r_i^{(c)})-k_i.
$$
Since the left-hand side is integral, it is at least $\lfloor\sum_cG_i(\tilde r_i^{(c)})\rfloor-k_i+1$. Combining the two cases,
$$
\sum_cF_i(\hat r_i^{(c)})
\geq
\left\lfloor\sum_cG_i(\tilde r_i^{(c)})\right\rfloor
-\max\{k_i-1,0\}.
$$
If $E:=\sum_i k_i$ and $a$ agents have $k_i>0$, delete all isolated vertices from the bipartite graph of non-bound entries and let $b\leq n$ be the number of remaining world vertices. For $E>0$, \Cref{lem:chores_extreme_point_forest} gives $E\leq a+b-1\leq a+n-1$; the case $E=0$ is immediate. Therefore,
$$
\sum_i\max\{k_i-1,0\}=E-a\leq n-1.
$$
By \Cref{lem:uniform_relaxed_capacity}, $\lfloor\sum_cG_i(\tilde r_i^{(c)})\rfloor\geq m$ for every agent $i$. Therefore,
$$
\sum_{c=1}^n\sum_{i=1}^nF_i(\hat r_i^{(c)})\geq nm-n+1.
$$

We now restate and prove the chores theorem.

\exactMMSChoresTheorem*

\begin{proof}
The preceding integral bound implies that some world $c^\star$ satisfies
$$
\sum_iF_i(\hat r_i^{(c^\star)})\geq m;
$$
otherwise every integral world capacity would be at most $m-1$, giving a total of at most $nm-n$.

For agent $i$, create $\hat r_i^{(c^\star)}$ easy slots adjacent to $[\eta_i]$ and $F_i(\hat r_i^{(c^\star)})-\hat r_i^{(c^\star)}$ unrestricted slots adjacent to all chores. We show that the resulting bipartite graph has a matching covering the chores. For nonempty $X\subseteq[m]$, let $q:=\min X$. Every unrestricted slot is adjacent to $X$, as is every easy slot of an agent with $\eta_i\geq q$. Hence,
$$
|N(X)|=\sum_iF_i(\hat r_i^{(c^\star)})
-\sum_{i:\eta_i<q}\hat r_i^{(c^\star)}.
$$
If $k$ is the largest index with $\eta_k<q$, the rounded prefix condition yields
$$
\sum_{i:\eta_i<q}\hat r_i^{(c^\star)}
=\sum_{i=1}^k\hat r_i^{(c^\star)}
\leq\eta_k\leq q-1;
$$
the sum is zero if no such $k$ exists. Since $X\subseteq\{q,\ldots,m\}$,
$$
|N(X)|\geq m-q+1\geq|X|.
$$
Hall's theorem gives the desired matching.

Assign every chore to the agent owning its matched slot. Agent $i$ receives at most $\hat r_i^{(c^\star)}$ chores through easy slots and at most $F_i(\hat r_i^{(c^\star)})-\hat r_i^{(c^\star)}$ through unrestricted slots. Therefore,
$$
v_i(A(i))\leq
\hat r_i^{(c^\star)}
+\bv_i\bigl(F_i(\hat r_i^{(c^\star)})-\hat r_i^{(c^\star)}\bigr)
\leq\MMS_i.
$$
First, undo the temporary reversal of the common rank indices. Then undo the normalization and the standard IDO reduction, and finally invoke \Cref{lem:zero_lower_value_reduction}. These steps return an exact MMS allocation for the original instance.

Finally, the algorithm runs in polynomial time. For bivalued costs, each agent's MMS can be computed by a standard count-based dynamic program over the numbers of easy and difficult chores. The remaining construction uses only polynomial-size linear programs, max-flow computations, and bipartite matching, together with the polynomial-time reductions described above.
\end{proof}

\end{document}